\documentclass[preprint,12pt]{elsarticle}

\usepackage{amssymb}
\usepackage{amsthm}

\usepackage{amsfonts}
\usepackage{amssymb}
\usepackage{amsthm}
\usepackage{amstext}
\usepackage{amsbsy}
\usepackage{amsopn}
\usepackage{amscd}
\usepackage{amsxtra}
\usepackage{algorithmic,algorithm}
\usepackage{bm}
\usepackage{xcolor}
\usepackage{diagbox}
\usepackage{matlab-prettifier}
\usepackage[normalem]{ulem}

\newcommand*\ve{\pmb e}
\newtheorem{theorem}{Theorem}

\definecolor{Simon}{rgb}{0, 0, 1}

\definecolor{CGreen}{rgb}{0, 1, 0}

\definecolor{SimonNew}{rgb}{1, 0, 0}

\definecolor{GC}{rgb}{1, 0, 1}

\journal{Computer Physics Communications}

\begin{document}

\begin{frontmatter}



\title{A \texttt{SPIRED} code for the reconstruction of spin distribution}


\author[inst1]{Simon Buchwald}
\author[inst2]{Gabriele Ciaramella}
\author[inst3]{Julien Salomon}
\author[inst4]{Dominique Sugny}

\affiliation[inst1]{organization={Department of Mathematics, Universität Konstanz},%Department and Organization
            addressline={Universitätsstr. 10}, 
            city={Konstanz},
            postcode={78464}, 
            country={Germany}}

\affiliation[inst2]{organization={MOX, Dipartimento di Matematica, Politecnico di Milano},%Department and Organization
            addressline={Piazza Leonardo da Vinci 32}, 
            city={Milano},
            postcode={20133}, 
            country={Italy}}

\affiliation[inst3]{organization={INRIA Paris, ANGE team},%Department and Organization
            addressline={2 rue Simone Iff}, 
            city={Paris},
            postcode={75589}, 
            country={France}}

\affiliation[inst4]{organization={Laboratoire Interdisciplinaire Carnot de Bourgogne (ICB), UMR 6303 CNRS–Université de Bourgogne},%Department and Organization
            addressline={9 Av. A. Savary, B.P. 47 870}, 
            city={Dijon Cedex},
            postcode={F-21078}, 
            country={France}}

\begin{abstract}
In Nuclear Magnetic Resonance (NMR), it is of crucial importance to have an accurate knowledge of the sample probability distribution corresponding to inhomogeneities of the magnetic fields.
An accurate identification of the sample distribution requires a set of experimental data that is sufficiently rich to extract all fundamental information. 
These data depend strongly on the control fields (and their number) used experimentally.
In this work, we present and analyze a greedy reconstruction algorithm, and provide the corresponding \texttt{SPIRED} code, for the computation of a set of control functions allowing the generation of data that are appropriate for the accurate reconstruction of a sample distribution. In particular, the focus is on NMR and the Bloch system with inhomogeneities in the magnetic fields in all spatial directions. Numerical examples illustrate this general study.
\end{abstract}

% ---- COMMENTED BY GABRIELE ---- 
% \begin{highlights}
% \end{highlights}
% ---- 


\end{frontmatter}


\section{Introduction}
\label{sec:intro}

Quantum Control (QC) is nowadays a well-recognized area of research~\cite{alessandrobook,PRXQuantumsugny,altafini2012,dong2010,BCSbook} with many applications ranging from magnetic resonance~\cite{levitt2013spin,glaser_training_2015,lapert_exploring_2012} and atomic and molecular physics~\cite{brif:2010,RMP:rotation,BEC2021,lapertferrini2012} to quantum technologies~\cite{glaser_training_2015,QT,kochroadmap}. Its goal is generally to design external control fields to perform quantum operations on the studied system. A severe limitation of QC comes from measurement processes which are much more difficult to account for than their classical counterpart. This explains that a majority of QC protocols are performed in an open-loop framework without any feedback from the experiment when applying the control. A good agreement between theory and experiment is achieved if all the parameters of the model system are perfectly known within a given range of precision. The values of such parameters can be estimated experimentally but can also be actively found by using specifically adapted controls. To this aim, different approaches using quantum features have been developed recently with success~\cite{Helstrom_review_1969,Vittorio_review_2004,Degen_review_2017}. Among others, we can mention inversion techniques~\cite{madaysalomon}, selective controls~\cite{Conolly_optimal_1986,Zhang_minimum_2015,Ansel_2021,Van_Damme_time_optimal_2018}, the maximization of quantum Fischer information~\cite{liu2017,yuan2017,Lin_optimal_2021,Liu_optimal_QM_review_2022,Lin_Application_2022} and the fingerprinting approach~\cite{ma2013,ansel2017}. Such methods allow one to estimate the value of the Hamiltonian parameter as well as its variation range. However, this latter is not the only interesting quantity and the probability distribution is also a key feature of the experimental sample. When controlling an ensemble of quantum systems, this distribution can be interpreted as the number of individual systems having a given value of the parameter. The distribution can have a simple form such as a Gaussian or a Lorentzian one. In this case, the identification is quite straightforward and can be done using standard techniques. However, the identification is much more difficult when the distribution has a complex structure with, e.g., several peaks.

In a previous work~\cite{spinpaper}, we introduced a Greedy Reconstruction Algorithm (GRA) to identify in a systematic way the probability distribution of one given Hamiltonian parameter. 
This was based on the framework presented in~\cite{BCS2021,madaysalomon}.
In particular, we focused on an ensemble of spin 1/2 particles in Nuclear Magnetic Resonance (NMR) subjected to an inhomogeneous radio-frequency magnetic field~\cite{kobzar:2008,levitt2013spin,lapertprl,lapert_exploring_2012,khanejaspin,bonnard},
where the algorithm was successfully applied to identify the distribution of the scaling factor corresponding to the sample inhomogeneity. 
Notice that a convergence analysis was only briefly sketched in~\cite{spinpaper}, without rigorous proof.
The goal of the present paper is to extend the work~\cite{spinpaper} from different points of view. First, we extend the GRA for the reconstruction of joint distributions of two distinct inhomogeneous Hamiltonian parameters. Second, we provide full MATLAB codes implementing our GRA and its optimized version (called OGRA) to find spin distribution. Such codes can be directly used to solve the problems presented in~\cite{spinpaper} and those investigated in this study. Third, we take also the opportunity of this paper to prove theoretical results covering also the ones only stated in~\cite{spinpaper}. 
As a result, this paper not only considers a more general problem than the one presented in~\cite{spinpaper}, but also provides a full MATLAB code and detailed and rigorous convergence analysis.

The paper is organized as follows. The identification problem of the spin distribution in NMR is presented in Sec.~\ref{sec:distribution}. The different variants of the greedy reconstruction algorithm are described in Sec.~\ref{sec:greedy}. Section~\ref{sec:code} is dedicated to the description of the structure of the code \texttt{SPIRED} and its use. A convergence analysis of the algorithm is provided in Sec.~\ref{sec:analysis}. Numerical results are presented in Sec.~\ref{sec:numerics}.  Conclusion and prospective views are drawn in Sec.~\ref{sec:conclusion}. Additional results are presented in~\ref{sec:appendixA}.

%WE SHOULD HAVE MAX AROUND 25 PAGES IN THIS FORMAT!!

%Novelties:
%\begin{itemize}
%    \item The \texttt{SPIRED} code that allows to %reproduce also the results of \cite{spinpaper}.
%    \item Extension of \cite{spinpaper} to two distributions (also offset).
%    \item Proofs of the (extended) theoretical results only stated in \cite{spinpaper}.
%\end{itemize}

%\cite{BCS2021,spinpaper}

%%-------------------------------
\section{Identification of spin distribution}\label{sec:distribution}

The framework of our \texttt{SPIRED} code is illustrated in a standard control problem in NMR, i.e. a spin ensemble subjected to inhomogeneous radio-frequency magnetic fields~\cite{kobzar:2008,skinner:2005,lapert_exploring_2012}. In a given rotating frame, each isochromat is characterized by a Bloch vector $\textbf{M}=[M_x,M_y,M_z]^\top$, evolving in time according to the equations
$$
\begin{cases}
\dot{M}_x=-\omega M_y+(1+\alpha)\omega_yM_z, \\
\dot{M}_y=\omega M_x-(1+\alpha)\omega_xM_z, \\
\dot{M}_z=(1+\alpha)\omega_xM_y-(1+\alpha)\omega_yM_x.
\end{cases}
$$
Notice that the components of $M$ satisfy $M_x^2+M_y^2+M_z^2=M_0^2$, with $M_0$ the equilibrium magnetization. Here, $\omega_x$ and $\omega_y$ are time-dependent controls corresponding to the components of the magnetic field along the $x$- and $y$- directions.
The parameters $\omega$ and $\alpha$ correspond to offset and control field inhomogeneities, respectively~\cite{levitt2013spin}. 
In standard experiments, we have $\frac{\omega}{2 \pi} \in [-20,20]$ Hz and $\alpha \in [-0.2,0.2]$.
For the purpose of this paper, we assume that the probability densities of $\omega$ and $\alpha$ are unknown.
The controls $\frac{\omega_x}{2\pi}$ and $\frac{\omega_y}{2\pi}$ are expressed in Hz. We consider a typical field amplitude $\omega_0$ that can be fixed, for instance, to $\omega_0=2\pi\times 100$~Hz. We introduce normalized coordinates as follows:
$$
u_x=2\pi\frac{\omega_x}{\omega_0};~u_y=2\pi\frac{\omega_y}{\omega_0};~t'=\frac{\omega_0}{2\pi}t;~\Delta=2\pi\frac{\omega}{\omega_0};\textbf{X}=\frac{\textbf{M}}{M_0}.
$$
In what follows, we omit the prime to simplify the notations. We deduce that the differential system can be expressed in normalized units as:
\begin{equation}\label{eq1}
\begin{cases}
\dot{x}=-\Delta y+(1+\alpha)u_yz \\
\dot{y}=\Delta x-(1+\alpha)u_xz \\
\dot{z}=(1+\alpha)u_xy-(1+\alpha)u_yx
\end{cases}
\end{equation}
with $x^2+y^2+z^2=1$.
The initial state of the dynamics for each spin is the thermal equilibrium point, i.e. $\textbf{X}_0=[0,0,1]^\top$. We consider a control time of the order of 100~ms, that corresponds to a normalized time $t_f'$ of the order of 10. 
%In the numerical simulations, we add the constraints $|u_x|\leq u_m$ and $|u_y|\leq u_m$ where $u_m$ is the maximum amplitude of each component.
The range of variation of the parameter $\Delta$ is $\Delta_0+2\pi [-0.2,0.2]$, where $\Delta_0$ is a frequency value that can be used to shift arbitrarily the interval.
For the purpose of this paper, we assume that $\Delta_0\geq0.4\pi$, meaning that $\Delta\geq0$.

The goal of our \texttt{SPIRED} code is to estimate simultaneously the distributions for the parameters $\alpha$ and $\Delta$ by designing specific controls $(u_x,u_y)$. We consider an ensemble of $N$ spins whose dynamics are governed by Eq.~\eqref{eq1}. 
We assume that the control amplitudes $(u_x,u_y)$ belong to the admissible set $\mathcal{U}=\{(u_x,u_y)\in\mathbb{R}^2\mid|u_x|\leq u_m, |u_y|\leq u_m \}$, where $u_m$ is the maximum amplitude of each component. A simple way to proceed can be described as follows. We consider that the system of $N$ spins is divided into $K_\Delta$ groups, and we associate with the $\ell$-th subgroup a certain value $\Delta_\ell$ and the corresponding probability $P_\star^\Delta(\ell) = \frac{N_{\Delta,\ell}}{N}$, $\ell=1,\dots,K_\Delta$, with $\sum_{\ell=1}^{K_\Delta}P_\star^\Delta(\ell)=1$. 
The probability $P_\star^\Delta(\ell)$ is unknown, which means that the number of elements $N_{\Delta,\ell}$ of each group is to be found. 
Similarly, for the parameter $\alpha$, we have $K_\alpha$ groups with the probabilities $P_\star^\alpha(\ell) = \frac{N_{\alpha,\ell}}{N}$, $\ell=1,\dots,K_\alpha$, with $\sum_{\ell=1}^{K_\alpha}P_\star^\alpha(\ell)=1$ to estimate.

This problem can be viewed as a natural extension of the work~\cite{spinpaper} and leads to the identification of two independent discrete distributions. 
However, this approach has two main drawbacks.
First, the two random variables $\Delta$ and $\alpha$ are assumed to be independent. 
This is a limitation when trying to reconstruct the two unknown distributions, since any possible correlation is a priori neglected.
Second, the final identification problem is nonlinear, since the product of the two distributions would appear. 
This is in contrast with the case of the reconstruction of one single distribution, where the identification problem is quadratic~\cite{spinpaper}.
For these reasons, rather than considering two independent distributions, we work directly with the joint distribution, i.e. 
the system of $N$ spins is divided into $K$ groups  and we associate to each subgroup a pair $(\alpha,\Delta)_\ell$ and the corresponding joint probability $P_\star(\ell) = \frac{N_\ell}{N}$, $\ell=1,\dots,K$, with $\sum_{\ell=1}^{K}P_\star(\ell)=1$. Now, the joint probability $P_\star(\ell)$ is unknown, namely the number of elements $N_\ell$ affected by the pairs $(\alpha,\Delta)_\ell$. 
This approach has the advantage of taking into account correlation effects and the final identification problem remains quadratic. It should be noted that these are acquired at the cost of an increase in the dimension of the unknown object(s), i.e. from two one-dimensional functions to a two-dimensional function.
%One has to remark that these are gained at the price of increasing the dimension of the unknown object(s): from two one-dimensional functions to one two-dimensional functions. 
Finally, we point out that two independent distributions can also be treated as a specific case of joint distributions.
%However, the case of the two independent distributions is not lost, but is well described by considering a joint distribution.
%Thus, if in the experimental setting the two random variable $\Delta$ and $\alpha$ are truly independent, this property can be fully reconstructed using a two-dimensional joint distribution.

Since we are dealing with an inverse problem, we need to define what quantities can be observed in an experimental setting. In NMR, only the first two coordinates of the magnetization vector can be directly measured.
We do not have accessed directly to the $z$ component due to the strong constant magnetic field applied along this direction~\cite{levitt2013spin}. 
We denote by $\textbf{Y}_{\textbf{u},(\Delta,\alpha)}(t)=[x(t),y(t)]^\top$ the projection of the Bloch vector onto the first two coordinates.
Here, the dependence on $\textbf{u}$ and $(\Delta,\alpha)$ has been explicitly mentioned.
The corresponding experimental realization of this controlled dynamic is obtained at $t=t_f$ and leads to $\textbf{Y}^{\textrm{exp}}_{\textbf{u}}(t_f)=[x^{\textrm{exp}}_{\textbf{u}}(t_f),y^{\textrm{exp}}_{\textbf{u}}(t_f)]^\top$, where $\textbf{Y}^{\textrm{exp}}_{\textbf{u}}(t_f)$ can be viewed as the average at time $t_f$ of the experimental measures of all the spins of the set subjected to the control $\textbf{u}$. The coordinates $x^{\textrm{exp}}_{\textbf{u}}$ and $y^{\textrm{exp}}_{\textbf{u}}$ are those of this measured magnetization vector.

The relation between the theoretical description of the dynamical system to the experimental outcome can be expressed as:
\begin{equation}\label{eqpstar}
\textbf{Y}^{\textrm{exp}}_{\textbf{u}}(t_f)=\sum_{\ell=1}^{K}P_\star(\ell)\textbf{Y}_{\textbf{u},(\Delta,\alpha)_\ell}(t_f),
\end{equation}
in which the two sides of the equation crucially depend on the control $\textbf{u}$. 

%A specific control protocol is not sufficient to identify the probability distribution $P_\star$, which generally requires the implementation of $K$ control processes with $K$ different control functions denoted $\textbf{u}_k$, $k=1,\cdots,K$. Note that in the optimized version of the GRA presented in Appendix~\ref{OGRA}, the number of controls can be different from $K$.

In general, one control protocol is not sufficient to obtain an appropriate identification of the unknown $P_\star$, but a set of $\widetilde{K}$ control processes with $\widetilde{K}$ different control functions denoted $\textbf{u}_k$, $k=1,\cdots,\widetilde{K}$, needs to be used. 
On the basis of the experimental outputs, a straightforward way to determine $P_\star$ is to solve the following minimization problem:
\begin{equation}\label{pbmin}
\min_{P\in \mathbb{P}}\sum_{k=1}^{K}\|\textbf{Y}^{\textrm{exp}}_{\textbf{u}_k}(t_f)-\sum_{\ell=1}^{K}P(\ell)\textbf{Y}_{\textbf{u}_k,(\Delta,\alpha)_\ell}(t_f)\|^2,
\end{equation}
where $\|\cdot\|$ denotes the standard Euclidean vector norm, and $\mathbb{P}$ is the convex and closed set of all the possible probability distributions $P$ that satisfy $P(\ell)\geq 0$ for $1\leq \ell\leq K$ and $\sum_{\ell=1}^{K}P(\ell)=1$. 
At this point, it is clear that a key ingredient of the accuracy of the identification process rests on the choice of a set of $\widetilde{K}$ controls $\textbf{u}_k$.
The identification of the number $\widetilde{K}$ of control functions is a difficult task.
The theoretical analysis presented in Sec.~\ref{sec:analysis} shows that the choice $\widetilde{K}=K$ is sufficient.
The GRA algorithm computes exactly $\widetilde{K}=K$ control fields. 
However, we will show that OGRA is capable of reducing (halving) the number $\widetilde{K}$ of control fields while guaranteeing an accurate identification.

Let us now rewrite \eqref{pbmin} in a form that we consider in our \texttt{SPIRED} implementation. We introduce a set $\Phi:=\{\phi_j\}_{j=1}^{K}$ of linearly independent functions $\phi_j:\{1,\dots,K\}\to \mathbb{R}$ such that $\mathbb{P} \subset {\rm span}(\Phi)$, where $\textrm{span}$ denotes the vector space generated by the functions. 
Expressing $P$ as $P(\ell)=\sum_{j=1}^{K}\beta_{j}\phi_j(\ell)$, the minimization problem \eqref{pbmin} becomes:
\begin{equation}\label{eqbeta}
\min_{\beta\in \widehat{\mathbb{R}}^K}\sum_{k=1}^{K}\|\textbf{Y}^{\textrm{exp}}_{\textbf{u}_k}(t_f)-\sum_{\ell,j=1}^{K}\beta_j\phi_j(\ell)\textbf{Y}_{\textbf{u}_k,(\Delta,\alpha)_\ell}(t_f)\|^2,
\end{equation}
where the vector $\beta=(\beta_j)_{j=1}^K$ is taken in $\widehat{\mathbb{R}}^K$, a subset of $\mathbb{R}^K$, so that $P=\sum_j\beta_j\phi_j$ is a probability distribution. 

We show in this study that GRA allows us to design a set of controls $\textbf{u}_k$ that makes \eqref{eqbeta} solvable and well conditioned. 
The algorithm is composed of two steps, namely an offline and an online steps. 
In the offline step, GRA computes the controls $\textbf{u}_k$.
In this step, only the theoretical model is needed without any experimental input. 
The derived controls are used in the online step in which the different magnetization vectors are measured and the minimization problem~\eqref{pbmin} is solved.
Note that the controls are the same for any probability distribution
to identify and only depend on the model system under study.
Finally, we point out that in our algorithms the duration of each control pulse is considered as a variable to be optimized together with its amplitude.
In particular, we assume that the controls are constant in time, i.e. $\textbf{u}(t)\equiv\textbf{u}\in\mathbb{R}^2$, and that we can freely choose the control time up to a fixed maximum value $t_f$.
Since the initial state is an equilibrium point, this is equivalent to turning on the control at a time $t\geq 0$. We show in Sec.~\ref{sec:analysis} that these hypotheses are sufficient for the different examples to identify the probability distributions.
The generality of GRA allows one to tackle this situation in a straightforward manner.

%%-------------------------------
\section{Greedy reconstruction algorithms}
\label{sec:greedy}
We present in this section the GRA in its classical form and in an optimized extension called optimized GRA (OGRA). 

GRA computes the controls $\textbf{u}_k$ and the corresponding control times $t_k$ by solving a sequence of fitting-step and discriminatory-step problems. 
The goal of the fitting step is to identify a defect of the system, namely a nontrivial kernel of a certain matrix $W$ introduced below, while the discriminatory step designs a new control which is aimed to correct this discrepancy and to eliminate the identified nontrivial kernel. 
The explicit formulation of GRA is presented in Alg.~\ref{algo:GRA} and is given in terms of the function $\textbf{h}^{(k)}$ defined by:
\begin{equation}\label{hk}
\textbf{h}^{(k)}(\beta,\textbf{u},t)=\sum_{\ell=1}^{K}\sum_{j=1}^k\beta_j\phi_j(\ell)\textbf{Y}_{\textbf{u},(\Delta,\alpha)_\ell}(t),
\end{equation}
for any $\beta$ in $\mathbb{R}^k$.
%Note that $\textbf{h}^{(k)}(\beta,\textbf{u})$ only depends on the $k$-first coordinates of the vector $\beta$.
\begin{algorithm}[t]
	\caption{Greedy Reconstruction Algorithm (GRA)}
	\begin{algorithmic}[1]\label{algo:GRA} 
	\begin{small}
		\REQUIRE A set of $K$ linearly independent functions $\Phi=\{\phi_1,\ldots,\phi_{K}\}$. 
		\STATE Compute the control $\textbf{u}_1$ and the control time $t_1$ by solving 
        \begin{equation}\label{eq: initialization}
        \max_{\substack{\textbf{u}\in\mathcal{U}\\ t\in[0,t_f]}} \|\textbf{h}^{(1)}(1,\textbf{u},t)\|^2,
        \end{equation}
		\FOR{ $k=1,\dots, K-1$ }
		\STATE \underline{Fitting step}: Find $\beta^k=(\beta^{k}_j)_{j=1,\dots,k}$ that solves 
        \begin{equation}\label{eq: fitting step}
		\min_{\beta \in \mathbb{R}^k}\sum_{m=1}^{k}\| \textbf{h}^{(K)}(\ve_{k+1},\textbf{u}_m,t_m)-\textbf{h}^{(k)}(\beta,\textbf{u}_m,t_m)\|^2,
		\end{equation}
		where $\ve_{k+1}$ is the $(k+1)$-th canonical vector in $\mathbb{R}^K$.
		\STATE \underline{Discriminatory step}: Find $\textbf{u}_{k+1}$ and $t_{k+1}$ that solves 
        \begin{equation}\label{eq: discriminatory step}
		\max_{\substack{\textbf{u}\in\mathcal{U}\\ t\in[0,t_f]}}\|\textbf{h}^{(K)}(\ve_{k+1},\textbf{u},t)-\textbf{h}^{(k)}(\beta^k,\textbf{u},t)\|^2.
		\end{equation}
		%\STATE Update $k \leftarrow k+1$.
		\ENDFOR
	\end{small}
	\end{algorithmic}
\end{algorithm}
Notice that the fitting step minimizes over the full space $\mathbb{R}^k$, meaning that $\sum_j\beta_j\phi_j$ does not have to be a probability distribution.
However, this is a restrictive condition. On the contrary, it allows the algorithm to find and correct more nontrivial kernels than might be necessary.

One main characteristic of GRA is that the set $\Phi$ and its order have to be fixed a-priori.
However, the choice and order of $\Phi$ can have a crucial impact on the outcome of the algorithm as shown in \cite[Sec. 5.3]{BCS2021}.
Hence, the idea of OGRA, which is stated in Algorithm~\ref{algo:OGRA}, is to make the algorithm independent of the choice and order of the set $\Phi$.
Additionally, it aims at avoiding the computation of unnecessary control functions.
This is achieved by two adaptations in GRA.
The first one is that in each step one does not only consider the next element (the next canonical vector $\ve_{k+1}$) in the set, but all remaining elements (the canonical vectors $\ve_{k+\ell}$ for all $1\leq\ell\leq K-k$) in parallel.
Hence, it is also possible to enlarge the set $\Phi$ (and thus enlarge $K$) by additional functions $\phi_k$ which do not have to be linearly independent.
In order to progressively remove linearly dependent functions in the set and to avoid scaling issues, all remaining basis elements are orthonormalized against the already selected ones after each iteration.
The second adaptation is the introduction of two tolerances $\textrm{tol}_1,\textrm{tol}_2>0$.
The first tolerance $\textrm{tol}_1$ is used as a stopping criterion.
\begin{algorithm}[H]
	\caption{Optimized Greedy Reconstruction Algorithm (OGRA)}
	\begin{small}
		\begin{algorithmic}[1]\label{algo:OGRA}
			\REQUIRE A set of $K$ functions $\Phi=\{\phi_1,\ldots,\phi_{K}\}$
			and two tolerances $\textrm{tol}_1>0$ and $\textrm{tol}_2>0$.
			\STATE Compute $\textbf{u}_1$, $t_1$ and the index $\ell_1$ by solving the initialization problem
			\begin{equation*}%\label{eq: initialization OGRA}
            \max_{\ell\in\{1,\ldots,K\}}\max_{\substack{\textbf{u}\in\mathcal{U}\\ t\in[0,t_f]}} \|\textbf{h}^{(1)}(\ve_\ell,\textbf{u},t)\|^2,
            \end{equation*}
			\STATE Swap $\phi_1$ and $\phi_{\ell_1}$ in $\Phi$, and set $k=1$, $\widetilde{K}=1$, and $f_{max} = \|\textbf{h}^{(1)}(\ve_\ell,\textbf{u}_1,t_1)\|^2$.
			\WHILE{ $k\leq K-1$ and \begin{small}$f_{max}>\textrm{tol}_1$\end{small} }
			\FOR{$\ell=k+1,\ldots,K$}
			\STATE Orthonormalize all elements $(\phi_{k+1},\ldots,\phi_K)$ with respect to $(\phi_1,\ldots,\phi_k)$, remove any that are linearly dependent and update $K$ accordingly.
			\STATE \underline{Fitting step}: Find $(\beta^{\ell}_j)_{j=1,\dots,k}$ that solve the problem
			\begin{equation*}%\label{eq: fitting step OGRA}
		      \min_{\beta \in \mathbb{R}^k}\sum_{m=1}^{k}\| \textbf{h}^{(K)}(\ve_{k+\ell},\textbf{u}_m,t_m)-\textbf{h}^{(k)}(\beta,\textbf{u}_m,t_m)\|^2,
		      \end{equation*}
			and set $f_\ell = \sum_{m=1}^{k} \begin{small}\| \textbf{h}^{(K)}(\ve_{k+\ell},\textbf{u}_m,t_m)-\textbf{h}^{(k)}(\beta^\ell,\textbf{u}_m,t_m)\|^2\end{small}$.
			\ENDFOR
			\IF {$\max_{\ell=k+1,\dots,K} f_\ell > {\rm tol}_2$ }
			\STATE Set $\ell_{k+1} = \textrm{arg}\max_{\ell=k+1,\dots,K} f_\ell$.
			\ELSE
			\STATE \underline{Extended discriminatory step}: Find $\textbf{u}_{k+1}$, $t_{k+1}$ and $\ell_{k+1}$ that solve
			\begin{equation*}%\label{eq: discriminatory step OGRA}
		      \max_{\ell\in\{k+1,\ldots,K\}}\max_{\substack{\textbf{u}\in\mathcal{U}\\ t\in[0,t_f]}}\|\textbf{h}^{(K)}(\ve_{k+\ell},\textbf{u},t)-\textbf{h}^{(k)}(\beta^\ell,\textbf{u},t)\|^2.
		      \end{equation*}
            \STATE Set $\widetilde{K}=\widetilde{K}+1$.
			\ENDIF
			\STATE Swap $\phi_{k+1}$ and $\phi_{\ell_{k+1}}$ in $\Phi$.
            \STATE Set $f_{max}=\|\textbf{h}^{(K)}(\ve_{k+\ell_{k+1}},\textbf{u}_{k+1},t_{k+1})-\textbf{h}^{(k)}(\beta^{\ell_{k+1}},\textbf{u}_{k+1},t_{k+1})\|^2$.
            \STATE Set $k = k+1$.
			\ENDWHILE
		\end{algorithmic}
	\end{small}
\end{algorithm}
The algorithm terminates if the function value in the initialization or any of the discriminatory steps (denoted by $f_\ell$ in Alg.~\ref{algo:OGRA}) is too small, thus not adding new information.
The second tolerance $\textrm{tol}_2$ is used to skip the computation of a new control field in the discriminatory step, if the minimum cost function value computed by the fitting step is not small enough.
If this function value is large, then there already exists a control function that discriminates between the two distributions $\phi_{k+\ell}$ and $\sum_{j=1}^\ell \beta^\ell_j\phi_j$.
Notice that setting $\textrm{tol}_2$ to a very small value is reasonable if the final identification problem is quadratic.
In this case, one can prove that a nonzero cost function value in the fitting step implies that one does not need to compute a new control for the corresponding set element (compare with \cite{BCS2021}).
However, if the final identification problem is not quadratic, then it can make sense to set $\textrm{tol}_2$ to a larger value.
In conclusion, the main adaptations of OGRA in lines 3 and 8-9 allow the algorithm to reduce the number of computed controls $\widetilde{K}$, meaning that $\widetilde{K}<K$.
On the other hand, as we mentioned before, GRA is designed to always compute exactly $\widetilde{K}=K$ controls.
The numerical implementation of GRA and OGRA is presented and discussed in the following sections.

\section{The \texttt{SPIRED} code}\label{sec:code}
\subsection{Structure of the code}
In this section, we provide a full list of all MATLAB functions contained in the \texttt{SPIRED} code.
Inside the \texttt{SPIRED} folder the user can find the \texttt{main} routine used to run the \texttt{SPIRED} code, as well as the routines that run GRA and OGRA, and that solve their sub-problems (see Tab. \ref{tab:GRAroutines}).
\begin{table}[t]
\centering
\begin{tabular}{ l l }
	\hline
	\quad \textbf{GRA routines} & \textbf{Description}  \\
	\hline 
	\quad \texttt{main} & Main function used to run the code.  \\  
	\quad \texttt{GRA} & Greedy reconstruction algorithm.\\
	\quad \texttt{OGRA} &    Optimized greedy reconstruction algorithm.\\
	\quad \texttt{discriminatory\_step} & Routine that solves the initialization and\\
	& discriminatory step problem using MATLAB's \\
    & \textit{fmincon}-solver.\\
	\quad \texttt{fitting\_step} & Routine that solves the fitting step problem.\\
	\quad \texttt{orthonormalize} & Routine that orthonormalizes all remaining \\
	&basis elements after each iteration of OGRA.\\
    \quad \texttt{SVD\_solver} & Routine that solves the fitting step problem\\
    &using the singular value decompostion (SVD).
\end{tabular}
\caption{Routines related to the greedy reconstruction algorithm.}
\label{tab:GRAroutines}
\end{table}
Additionally, the \texttt{SPIRED} folder contains the routines that generate the synthetic experimental data for the true parameter probability distribution, and that solve the final identification problem \eqref{pbmin} (see Tab.~\ref{tab:Reconroutines}).
\begin{table}[t]
	\centering
	\begin{tabular}{ l l }
		\hline
		\quad \textbf{Reconstruction routines} & \textbf{Description}  \\
		\hline
		\quad \texttt{generate\_data}& Routine that generates the experimental \\
		& realizations for all computed controls.\\
		\quad \texttt{reconstr} & Routine that either solves the final \\
        & identification problem \eqref{eqbeta} using a second\\
        & order interior-point algorithm, or the\\
		& compact form (compare \eqref{eq: pbminW} in Section \ref{sec:analysis}) \\
		& using a solver based on the SVD.
	\end{tabular}
 \caption{Routines related to the reconstruction of the probability distribution.}
 \label{tab:Reconroutines}
\end{table}
Notice that the both the \texttt{discriminatory\_step} and the \texttt{reconstr} routine use MATLAB's \textit{fmincon}-solver, which requires MATLAB's Optimization Toolbox to be installed.

There are also three subfolders labeled ``Test1'', ``Test2'' and ``Test3''.
These contain three test problems the user can choose from.
``Test1'' corresponds to the problem discussed in this paper.
``Test2'' is the same as ``Test1'', but only considers a control in the x direction (in other words $\textbf{u}_y=0$ for all control fields).    
Finally, ``Test3'' corresponds to the problem investigated in \cite{spinpaper}, where the resonance offset $\Delta$ is fixed and one attempts to reconstruct only the control inhomogeneity parameter $\alpha$.
Each of these ``Test'' folders contains routines to set the input variables, describe the cost function and its gradient for the discriminatory-step problem, and solve the corresponding dynamical system (see Tab.~\ref{tab:Testroutines}).
\begin{table}[t]
	\centering
	\begin{tabular}{ l l }
        \hline
		\quad \textbf{Test routines} & \textbf{Description}  \\
		\hline
		\quad \texttt{starting} & Input function.\\
		\quad \texttt{fun\_discriminatory} & Function computing the cost functional and \\
		& gradient for the discriminatory step problem.\\
		\quad \texttt{NMR\_solver} & Routine that solves the (normalized) dynamical \\
		& system via direct calculations of the exponential\\
		& matrix (compare the proof of Theorem \ref{thm: Positive Discriminatory Step}).\\
	\end{tabular}
 \caption{Routines related to the test problems.}
 \label{tab:Testroutines}
\end{table}
They also each contain two routines used to plot the results and condition number of the reconstruction process (see Tab.~\ref{tab:Plotroutines}).
\begin{table}[t]
	\centering
	\begin{tabular}{ l l }
		\hline
		\quad \textbf{Plotting routines} & \textbf{Description}  \\
		\hline
		\quad \texttt{plot\_reconstr} & Routine that plots the true and reconstructed \\
		& probability distributions for the two control sets.\\
		\quad \texttt{plot\_condition} & Routine that produces a table containing the\\
        & condition numbers corresponding to the \\
        & reconstruction process.
	\end{tabular}
 \caption{Routines plotting the results for the test problem.}
 \label{tab:Plotroutines}
\end{table}

\subsection{Usage of the code}
Here we illustrate the working procedure of the \texttt{SPIRED} code with an example.
The user needs to define the test problem in the function \texttt{main}, which is used to initialize the procedure.
\begin{lstlisting}[style=Matlab-editor,basicstyle=\small]
function [ controls, bases, model, results ] = main
% STEP 1: Choice of the Problem;
addpath('Test1');
% STEP 2: Assemble problem variables;
[ model, bases, options ] = starting ( );
% STEP 3: Run!
[controls.GRA, results.GRA] = GRA( bases.GRA, model, options);
[controls.OGRA, bases.OGRA, results.OGRA] = OGRA( bases.OGRA, model, options);
% STEP 4: Compute (synthetic) experimental data;
Y_exp.GRA  = generate_data( controls.GRA, model );
Y_exp.OGRA = generate_data( controls.OGRA, model );
% STEP 5: Solve the final identification problem; 
...
\end{lstlisting}
In particular, at the ``STEP 1'' the user needs to define the path of the folder containing the test routines.
Then, the user can define the input variables in the function \texttt{starting}, which is listed exemplary for the first test problem in the following.
\begin{lstlisting}[style=Matlab-editor,basicstyle=\small]
function [ model, bases, options ] = starting ( )
% STEP 1: Input variables;
% control bounds and maximum control time;
um  = 10;
tf  = 16;
% variables for the unknown parameters
Delta0          = 4*pi;
Delta1          = 0.2;
Delta_interval  = Delta0 + 2*pi.*[-Delta1, Delta1];
alpha_interval  = [-0.2, 0.2];
% number of grid points for the unknown parameters
nr_alphas       = 10;
nr_Deltas       = 10;
% open the file to get the input probability distribution
load('Test1/Distributions/Gaussian.mat', 'P_star')
% number of spins;
nr_spins = 100000;
% options for GRA and OGRA
iterations      = nr_alphas*nr_Deltas;
Display_GRA     = 'off';
flag_orth       = 1;
% numerical parameters for OGRA;
tol_OGRA_fit    = 1e-4;
tol_OGRA_discr  = 1e-14;
% tolerance for the SVD solver in the fitting step (and optionally for the reconstruction solver)
tol_svd         = 1e-10;
% optimization method for the final identification problem
solver          = 'fmincon';
\end{lstlisting}
At ``STEP 1'' in this function, the user can define the input variables and the path to the .mat file containing the true probability distribution $P_\star$.
The input parameters related to the problem are
\begin{itemize}\itemsep0em
	\item \texttt{um:} bound $u_m$ for the absolute value of the control amplitudes;
	\item \texttt{tf:} maximum normalized control time $t_f$;
	\item \texttt{Delta0:} frequency shift $\Delta_0$ for the normalized resonance offset interval;
	\item \texttt{Delta1:} width $\Delta_1$ of the normalized resonance offset interval;
	\item \texttt{Delta\_interval:} interval boundaries for the normalized resonance offset $\Delta$;
	\item \texttt{alpha\_interval:} interval boundaries for the control field inhomogeneity parameter $\alpha$;
	\item \texttt{nr\_alphas:} number of grid points in the direction of $\alpha$ for the joint discrete parameter probability distribution of $\alpha$ and $\Delta$;
	\item \texttt{nr\_Deltas:} number of grid points in the direction of $\Delta$ for the joint discrete parameter probability distribution of $\alpha$ and $\Delta$;
    \item \texttt{nr\_spins:} number of spins in the system;
	\item \texttt{iterations:} (maximum) number of iterations performed by GRA and OGRA; for any full basis of the discrete parameter space, the obvious choice is the total number of discretization points, which is the product of \texttt{nr\_alphas} and \texttt{nr\_Deltas};
	\item \texttt{Display\_GRA} Display option to print information about the current iteration of GRA and OGRA in the command window; can be set to \texttt{'off'} to display no output, \texttt{'iter'} to show the current substep of GRA and OGRA, or \texttt{'iter-detailed'} to also show the current optimization problem during the substeps of OGRA;
	\item \texttt{flag\_orth:} flag variable that turns the orthonormalization of the remaining basis elements during OGRA on or off;
	\item \texttt{tol\_OGRA\_fit:} tolerance $\textrm{tol}_2$ for OGRA;
	\item \texttt{tol\_OGRA\_discr:} tolerance $\textrm{tol}_1$ for OGRA;
    \item \texttt{tol\_svd:} tolerance for the SVD solver, used in the fitting step and (optionally) for the final identification problem;
    \item \texttt{solver:} optimization method used to solve the final identification problem \eqref{eqbeta}; can be set to \texttt{'fmincon'} to solve \eqref{eqbeta} using the second-order interior-point algorithm of MATLAB's \textit{fmincon}-solver, or to \texttt{'svd'} to solve a compact form of the problem (compare \eqref{eq:Reconstruction} in Section \ref{sec:numerics}) using the SVD solver;
\end{itemize}
The .mat file has to contain the variable \texttt{P\_star}, which is the vectorized true joint probability distribution $P_\star$.
If the user is considering a true experimental (laboratory) setup, meaning that they perform real experiments for the different controls to obtain the experimental data and that the true probability distribution is truly unknown, they should replace ``STEP 4'' in the "main.m" file with a load command to fetch the real experimental data.

Finally, to run the code the user has to write on the MATLAB prompt the following
\begin{lstlisting}[style=Matlab-editor,basicstyle=\small]
>> [ controls, bases, model, results ] = main
\end{lstlisting}
After the computations, the routine saves the results in the MATLAB work\-space (as documented in the code) and plots the reconstructed probability distributions and their difference to the true one, as well as the condition numbers for different mesh sizes.

In particular, the results obtained by running ``Test1'' are the true and reconstructed probability distributions for the control fields generated by GRA and OGRA, shown in Fig.~ \ref{fig:gauss_reconstructed}.
\begin{figure}[t]
	\centering
    \includegraphics[width=1\linewidth]{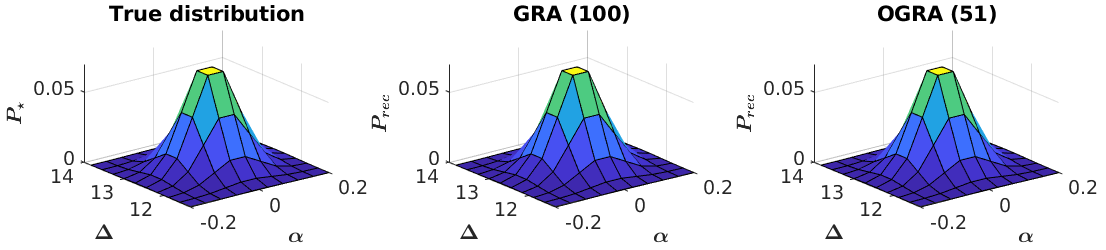}
    \caption{The plot on the left shows the true Gaussian probability distribution for $K=100$ uniform mesh points. The plots in the middle and on the right contain the reconstructed distributions for the control sets generated by GRA (containing $100$ control fields) and OGRA (containing $51$ control fields).}
    \label{fig:gauss_reconstructed}
\end{figure}
In Fig.~\ref{fig:Test1_err} we show the difference between the true and reconstructed distributions for the two control field sets.
\begin{figure}[t]
	\centering
	\includegraphics[width=0.9\linewidth]{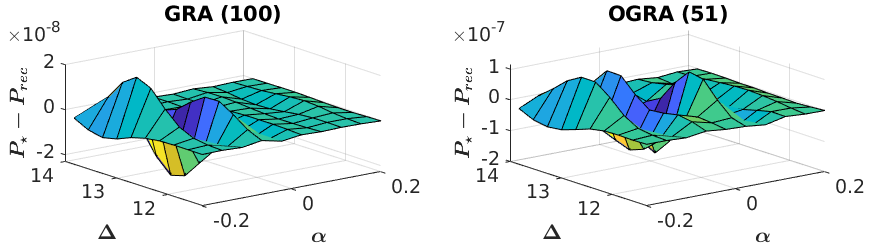}
	\caption{Difference between the true probability distribution $P_\star$ and the distributions $P_{rec}$ reconstructed using the control sets generated by GRA (left) and OGRA (right). In brackets are the number of control fields for each set.}
	\label{fig:Test1_err}
\end{figure}
Additionally, the routine provides a table containing the exact condition numbers corresponding to GRA and OGRA.
Since the solver for the discriminatory step problem is initialized with a random vector, there may be small variations in some results, without changing the overall outcome.

Examples of all figures produced by the different test problems are also provided in the folder ``Results'' that can be found in the corresponding ``Test''-folder.
There one can also find the .mat files containing the set of random controls for each test problem, loaded in the \texttt{main}.

\section{Convergence analysis}\label{sec:analysis}
In this section, we prove that the controls generated by GRA and OGRA make possible the identification of the unknown probability distributions of the parameters $\Delta$ and $\alpha$, i.e. they make problem \eqref{pbmin} uniquely solvable.

We start by recalling that problem \eqref{pbmin} is equivalent to \eqref{eqbeta}.
Assuming that $P_\star$ can be written as $P_\star(\ell)=\sum_{j=1}^{K}\beta_{\star,j}\phi_j(\ell)$, we can write equation~\eqref{eqbeta} in a compact form as follows:
\begin{equation}\label{eq: pbminW}
\min_{\beta\in \widehat{\mathbb{R}}^K}\langle \beta_\star-\beta|W|\beta_\star-\beta\rangle,
\end{equation}
where $W:=\sum_kW(\textbf{u}_k,t_k)$ is the sum of symmetric and positive semi-definite $K\times K$- matrices whose elements are defined as:
\begin{equation}\label{eq: W}
[W(\textbf{u}_k,t_k)]_{\ell,j} :=\langle \gamma_{\ell}(\textbf{u}_k,t_k)|\gamma_j(\textbf{u}_k,t_k)\rangle
\end{equation}
with
\begin{equation}\label{eq: gamma}
\gamma_j(\textbf{u}_k,t_k):=\sum_\ell\phi_j(\ell)\textbf{Y}_{\textbf{u}_k,(\alpha_\ell,\Delta_\ell)}(t_k).
\end{equation}
Since the set of vectors $\beta$ is a convex subset of $\mathbb{R}^{K}$, we deduce that the problem is uniquely solvable if the matrix $W$ is positive definite. In the case $W$ has a non-trivial kernel, infinitely many solutions may exist which lead to wrong probability distributions different from the experimental one $P_\star$. We stress that the non-triviality of the kernel depends completely on the choice of the controls $\textbf{u}_k$ and the corresponding control times $t_k$. 

Using the notation \eqref{eq: W}-\eqref{eq: gamma}, we can now also rewrite the subproblems of GRA in terms of the matrix $W$. The initialization problem \eqref{eq: initialization} can be written as
\begin{equation}\label{eq: initialization W}
\max_{\substack{\textbf{u}\in\mathcal{U}\\ t\in[0,t_f]}} |[W(\textbf{u},t)]_{1,1}|^2.
\end{equation}
The fitting step problem \eqref{eq: fitting step} is equivalent to
	%\begin{equation*}
	%\min_{\beta \in \mathbb{R}^k}\langle\beta|[W^k]_{[1:k,1:k]}|\beta \rangle-2\langle[W^k]_{[1:k,k+1]}|\beta\rangle,
	%\end{equation*}
\begin{equation}\label{eq: fitting step W}
\min_{\beta \in \mathbb{R}^k}\langle\textbf{v}_{\beta}|[W^k]_{[1:k+1,1:k+1]}|\textbf{v}_{\beta}\rangle,
\end{equation}
where $W^k:=\sum_{m=1}^{k}W(\textbf{u}_m,t_m)$ and $\textbf{v}_{\beta}:=[\beta^\top,-1]^\top$.
Finally, the discriminatory step problem \eqref{eq: discriminatory step} can be written as
\begin{equation}\label{eq: discriminatory step W}
\max_{\substack{\textbf{u}\in\mathcal{U}\\ t\in[0,t_f]}}\langle\textbf{v}_{\beta^k}|[W(\textbf{u},t)]_{[1:k+1,1:k+1]}|\textbf{v}_{\beta^k}\rangle.
\end{equation}
A direct interpretation of these reformulated problems is that each control $\textbf{u}_k$ generated by GRA at iteration $k$ ensures that $\langle \ve_k|W|\ve_k\rangle>0$.
Iteratively, this implies that $\langle\textbf{w}|W|\textbf{w}\rangle>0$ for any $\textbf{w}\in\mathbb{R}^K$ which is equivalent to $W$ being positive definite.

In more details, the first control $\textbf{u}_1$ and the control time $t_1$ are chosen by the initialization \eqref{eq: initialization W} such that the first upper left entry of $W(\textbf{u}_1,t_1)$ is positive.
This guarantees that $\langle \ve_1|W|\ve_1\rangle>0$ since
$$\langle \ve_1|W|\ve_1\rangle=\sum_{j=1}^K\langle \ve_1|W(\textbf{u}_j,t_j)|\ve_1\rangle\geq [W(\textbf{u}_1,t_1)]_{1,1}>0, $$
where we used that $W(\textbf{u},t)$ is positive semi-definite for any $\textbf{u}$ and $t$.
Assume now that the upper left $2\times2$-submatrix of $W^1=W(\textbf{u}_1,t_1)$ is not positive definite. 
Then it has a one-dimensional kernel spanned by a vector $\textbf{v}_{\beta^1}:=[\beta^1,-1]^\top\in\mathbb{R}^2$ (see \cite[Lem. 5.3]{BCS2021}).
The corresponding scalar $\beta^1$ is clearly the unique solution to the fitting step problem \eqref{eq: fitting step W} for $k=1$.
Now, the discriminatory step problem \eqref{eq: discriminatory step W} attempts to find a control $\textbf{u}_2\in\mathcal{U}$ and a control time $t_2\in[0,t_f]$ such that the vector $\textbf{v}_{\beta^1}$ is not in the kernel of the upper left $2\times2$-submatrix of $W(\textbf{u}_2,t_2)$.
If this is successful then the upper left $2\times2$-submatrix of $W^2=W(\textbf{u}_1,t_1)+W(\textbf{u}_2,t_2)$ is positive definite.
This also implies that $\langle \ve_2|W|\ve_2\rangle>0$.
Repeating this procedure for $k=2,\ldots,K-1$, we obtain $\langle \ve_k|W|\ve_k\rangle>0$ for all $k\in{1,\ldots,K}$, which guarantees that $W$ is positive definite.
We summarize the arguments above in the following theorem.
\begin{theorem}\label{thm:uniqsolution}
    Let $\{(\textbf{u}_k,t_k)\}_{k=1}^K$ be a set of controls and corresponding control times generated by GRA, such that $[W(\textbf{u}_1,t_1)]_{1,1}>0$. Let $\beta^k$ be the solution to the fitting step problem \eqref{eq: fitting step W} for $k=1,\ldots,K-1$, such that the vectors $\textbf{v}_{\beta^k}=[(\beta^k)^\top,-1]^\top$ are not in the kernel of $[W(\textbf{u}_{k+1},t_{k+1})]_{[1:k+1,1:k+1]}$.
    Then the matrix $W=\sum_kW(\textbf{u}_k,t_k)$ is positive definite.
\end{theorem}
It remains to show that the discriminatory step can always find a control such that the vector $\textbf{v}_{\beta^k}$ is not in the kernel of $[W(\textbf{u},t)]_{[1:k+1,1:k+1]}$.
In fact, it is sufficient to show that for any $k\in\{1,\ldots,K\}$ there exists a control $\textbf{u}\in\mathcal{U}$ and a $t\in[0,t_f]$ such that $\langle\textbf{v}_{\beta^k}|[W(\textbf{u},t)]_{[1:k+1,1:k+1]}|\textbf{v}_{\beta^k}\rangle>0$.
We show in Theorem~\ref{thm: Positive Discriminatory Step} that this is valid in the context of this paper.
\begin{theorem}\label{thm: Positive Discriminatory Step}
    Let $k\in\{1,\ldots,K-1\}$, $W^k_{[1:k,1:k]}$ be positive definite, $\beta^k$ the solution to the fitting-step problem~\eqref{eq: fitting step W}, and $\textbf{v}_{\beta^k} =[(\beta^k)^\top,-1]^\top$. Then any solution $(\textbf{u},t)$ to the discriminatory-step problem~\eqref{eq: discriminatory step W} satisfies
	\begin{equation*}
	\langle\textbf{v}_{\beta^k}|W_{[1:k+1,1:k+1]}(\textbf{u},t)|\textbf{v}_{\beta^k}\rangle
	=\|\textbf{h}^{(K)}(\ve^{k+1},\textbf{u};t) - \textbf{h}^{(k)}(\beta^k,\textbf{u};t)\|^2>0.
	\end{equation*}
\end{theorem}
\begin{proof}
    For brevity, we identify $\alpha$ with $1+\alpha$ for the remainder of this proof.
    We start by writing
	\begin{equation*}
	\textbf{h}^{(K)}(\ve^{k+1},\textbf{u},t) - \textbf{h}^{(k)}(\beta^k,\textbf{u},t)=\sum_{\ell=1}^K\Bigg( \phi_{k+1}(\ell)-\sum_{j=1}^k \beta^k_j \phi_j(\ell)\Bigg) \textbf{Y}_{\textbf{u},(\Delta,\alpha)_\ell}(t).
	\end{equation*}
	Since the functions $\{ \phi_1,\ldots,\phi_K \}$ are linearly independent, it holds that
	\begin{equation}\label{eq:widetildeell2}
	\exists \widetilde{\ell}\in\{1,\ldots,K\}:\quad\phi_{diff}(\widetilde{\ell}):= \phi_{k+1}(\widetilde{\ell})- \sum_{j=1}^k \beta^k_j \phi_j(\widetilde{\ell})\neq 0.
	\end{equation}
    According to \eqref{eq1}, we have for any $(\alpha,\Delta)_{\ell}$
	\begin{equation}\label{eq: Matrix ODE in X}
	\frac{d}{dt}\textbf{X}(t)=\Big[\Delta_\ell A+\alpha_{\ell} (u_x B_x+u_y B_y)\Big]\textbf{X}(t),\quad
	\textbf{X}(0)=\textbf{X}_0,
	\end{equation}
	where
	\begin{equation*}
	A=\begin{bmatrix}
	0&-1&0\\
	1&0&0\\
	0&0&0
	\end{bmatrix},\quad
	B_x=\begin{bmatrix}
	0&0&0\\
	0&0&-1\\
	0&1&0
	\end{bmatrix},\quad
	B_y=\begin{bmatrix}
	0&0&1\\
	0&0&0\\
	-1&0&0
	\end{bmatrix},\quad
	\textbf{X}_0=\begin{bmatrix}
	0\\
	0\\
	1
	\end{bmatrix}.
	\end{equation*}
    Now, consider the control $\widetilde{\textbf{u}}:=[0,b]^\top$ and a corresponding control time $\widetilde{t}\in[0,t_f]$, where both $b\in\mathbb{R}\setminus\{0\}$ and $\widetilde{t}$ are to be chosen later.
    We have $\textbf{Y}_{\textbf{u},(\Delta,\alpha)_\ell}(\widetilde{t})=C\textbf{X}(\textbf{u},(\alpha,\Delta)_{\ell};\widetilde{t})$, where $\textbf{X}(\textbf{u},(\alpha,\Delta)_{\ell};\widetilde{t})$ is the solution to \eqref{eq: Matrix ODE in X} and $C=\begin{bmatrix}
	1&0&0\\0&1&0
	\end{bmatrix}$. 
    Thus, we obtain $$\textbf{Y}_{\widetilde{\textbf{u}},(\Delta,\alpha)_\ell}(\widetilde{t})=Ce^{\widetilde{t}(\Delta_\ell A+\alpha_\ell bB_y)}\textbf{X}_0.$$
	Since $\Delta_\ell A+\alpha_\ell bB_y$ is skew-symmetric, we can compute its exponential matrix explicitly.
	By setting $\widetilde{A}:=\widetilde{t}(\Delta_\ell A+\alpha_\ell bB_y)$ and $x_\ell:=\sqrt{\Delta_\ell^2+\alpha_{\ell}^2b^2}$, we have
	$e^{\widetilde{A}}=I_3+\frac{\sin(\widetilde{t}x)}{\widetilde{t}x}\widetilde{A}+\frac{1-\cos(\widetilde{t}x_\ell)}{\widetilde{t}^2x_\ell^2}\widetilde{A}^2$ (see, e.g., \cite{Rodrigues1840}).
	Since $CI_3\textbf{X}_0=0$ and $$\;\widetilde{A}^2=\widetilde{t}^2\begin{bmatrix}
	-\Delta_\ell^2-\alpha_{\ell}b^2&0&0\\
	0&-\Delta_\ell^2&\Delta_\ell\alpha_{\ell}b\\
	0&\Delta_\ell\alpha_{\ell}b&-\alpha_{\ell}^2b^2
	\end{bmatrix},$$
	we obtain
 \begin{equation*}
     \resizebox{.95\hsize}{!}{$C\textbf{X}(\widetilde{\textbf{u}},\alpha_{\ell};\widetilde{t})=C\Bigg(\frac{\sin(\widetilde{t}x_\ell)}{x_\ell}\begin{bmatrix}
	\alpha_{\ell}b\\0\\0
	\end{bmatrix}+\frac{1-\cos(\widetilde{t}x_\ell)}{x_\ell^2}\begin{bmatrix}
	0\\\Delta_\ell\alpha_{\ell}b\\-\alpha_{\ell}^2b^2
	\end{bmatrix} \Bigg)
    =\begin{bmatrix}
	\frac{\sin(\widetilde{t}x_\ell)}{x_\ell}\alpha_{\ell}b\\
	\frac{1-\cos(\widetilde{t}x_\ell)}{x_\ell^2}\Delta_\ell\alpha_{\ell}b
	\end{bmatrix}.$}
 \end{equation*}
	Thus, we have
    \begin{equation*}
	\textbf{h}^{(K)}(\ve^{k+1},\widetilde{\textbf{u}},\widetilde{t}) - \textbf{h}^{(k)}(\beta^k,\widetilde{\textbf{u}},\widetilde{t})=b\sum_{\ell=1}^{K}\phi_{diff}(\ell)\alpha_{\ell}\begin{bmatrix}
	\frac{\sin(\widetilde{t}x_\ell)}{x_\ell}\\
	\Delta_\ell\frac{1-\cos(\widetilde{t}x_\ell)}{x_\ell^2}
	\end{bmatrix}=:F(\widetilde{t}).
    \end{equation*}
	Now, seeking a contradiction, assume that $\textbf{h}^{(K)}(\ve^{k+1},\widetilde{\textbf{u}},\widetilde{t}) - \textbf{h}^{(k)}(\beta^k,\widetilde{\textbf{u}},\widetilde{t})=0$ for all $\widetilde{t}\in[0,t_f]$ and all $b\in\mathbb{R}\setminus\{0\}$.
	Since $F$ is analytic in $\widetilde{t}$, we obtain $F^{(k)}(\widetilde{t})=0$ for all $k\in\mathbb{N}$ and all $\widetilde{t}\in[0,t_f]$.
	For $k$ odd, we have
	\begin{equation}\label{eq:FkT2}
	F^{(k)}(\widetilde{t})=b\sum_{\ell=1}^{K}\phi_{diff}(\ell)\alpha_{\ell}(-1^{\frac{k-1}{2}})\begin{bmatrix}
	x_\ell^{k-1}\cos(Tx_\ell)\\\Delta_\ell x_\ell^{k-2}\sin(Tx_\ell))
	\end{bmatrix}.
	\end{equation}
	Since $F^{(k)}(T)=0$ for all $k$ odd, the first component of $F^{(k)}(T)$ in \eqref{eq:FkT2}, for different $k$ odd, implies that
	\begin{equation*}
	\underbrace{\begin{bmatrix}
		1&1&\cdots& 1\\
		x_1^2&x_2^2&\cdots&x_K^2\\
		x_1^4&x_2^4&\cdots&x_K^4\\
		\vdots&\vdots&\vdots&\vdots\\
		x_1^{\widetilde{K}}&x_2^{\widetilde{K}}&\cdots&x_K^{\widetilde{K}}\\
		\end{bmatrix}}_{=:D}\underbrace{\begin{bmatrix}
		\phi_{diff}(1)\alpha_{1}\cos(Tx_1)\\
		\phi_{diff}(2)\alpha_{2}\cos(Tx_2)\\
		\vdots\\
		\phi_{diff}(K)\alpha_{K}\cos(Tx_K)
		\end{bmatrix}}_{=:\pmb{\phi}_{\widetilde{t}}}=0.
	\end{equation*}
	Notice that $D$ is a Vandermonde matrix (see, e.g., \cite{Lundengrd2017GeneralizedVM}).
	Now, let $\widetilde{K}=2(K-1)$, meaning that $D\in\mathbb{R}^{(\frac{\widetilde{K}}{2}+1)\times K}$ is a square matrix.
	Then, the determinant of $D$ is given exactly by 
	$$\det(D)=\prod_{1\leq i<j\leq K}(x_j^2-x_i^2).$$
	This implies that two rows of $D$ are linearly independent if and only if $|x_i|\neq|x_j|$.
	Hence, $\det(D_x)=\det(D_y)\neq0$ (and therefore $\pmb{\phi}_{\widetilde{t}}=0$) if and only if $|x_i|\neq|x_j|$ for $i\neq j$.
	Recalling that $x_\ell=\sqrt{\Delta_\ell^2+\alpha_{\ell}^2b^2}$, $|x_i|\neq|x_j|$ is equivalent to $\Delta_i^2+\alpha_{i}^2b^2\neq\Delta_j^2+\alpha_j^2b^2$.
    For $i\neq j$ we also have $\alpha_i\neq\alpha_j$ and/or $\Delta_i\neq\Delta_j$ by definition.
	Since $\alpha_\ell\in[0.8,1.2]$ and $\Delta_\ell\in \Delta_0+2\pi[-0.2,0.2]$ with $\Delta_0\geq0.4\pi$, we obtain $\alpha_i^2\neq\alpha_j^2$ and/or $\Delta_i^2\neq\Delta_j^2$ for $i\neq j$.
    Thus, there exists $b\in\mathbb{R}\setminus {0}$ such that $\Delta_i^2+\alpha_{i}^2b^2\neq\Delta_j^2+\alpha_j^2b^2$ for all $i,j\in\{1,\ldots,K\}$ with $i\neq j$.
    In conclusion, we have $|x_i|\neq|x_j|$ for $i\neq j$, which implies that $\pmb{\phi}_{\widetilde{t}}=0$ and therefore $\phi_{diff}(\ell)\alpha_{\ell}\cos(\widetilde{t}x_\ell)=0$ for all $\ell\in\{0,\ldots,K\}$ and all $\widetilde{t}\in[0,t_f]$.
	However, we also have $\phi_{diff}(\widetilde{\ell})\neq0$ by \eqref{eq:widetildeell2}, $\alpha_{\widetilde{\ell}}>0$ and $x_{\widetilde{\ell}}>0$.
	Thus, there exists $\widetilde{t}\in[0,t_f]$ such that $\phi_{diff}(\widetilde{\ell})\alpha_{\widetilde{\ell}}\cos(\widetilde{t}x_{\widetilde{\ell}})\neq0$, which is a contradiction.
\end{proof}
Analogously to the proof of Theorem \ref{thm: Positive Discriminatory Step}, one can show that any solution $(\textbf{u}_1,t_1)$ to the initialization problem \eqref{eq: initialization W} satisfies $[W(\textbf{u}_1,t_1)]_{1,1}>0$. 
We conclude our analysis by the following theorem.
\begin{theorem}
    Let $(\textbf{u}_k,t_k)$, $k=1,\ldots,K$, be a set of controls and corresponding control times generated by GRA.
    Then problem~\eqref{eqbeta} is uniquely solvable by $\beta=\beta_\star$.
\end{theorem}
\begin{proof}
    Let $\beta^k$ be the solution to the fitting step problem \eqref{eq: fitting step W} for $k=1,\ldots,K-1$. By Theorem \ref{thm: Positive Discriminatory Step}, the vector $\textbf{v}_{\beta^k} =[(\beta^k)^\top,-1]^\top$ is not in the kernel of $[W(\textbf{u}_{k+1},t_{k+1})]_{[1:k+1,1:k+1]}$ for all $k\in\{1,\ldots,K-1\}$.
    Thus, we obtain by Theorem~\ref{thm:uniqsolution} that the matrix $W=\sum_kW(\textbf{u}_k,t_k)$ is positive definite.
    Hence, problem~\eqref{eq: pbminW} is uniquely solvable by $\beta=\beta_\star$.
    By equivalency of problems~\eqref{eq: pbminW} and \eqref{eqbeta}, we obtain the result.
\end{proof}

Notice that, in the notation above, OGRA simply reorders rows and columns of the matrix $W^k$ while attempting to find and correct its kernel.
In fact, the second improvement in lines 8-9 in OGRA skips the discriminatory step only if there exists a row and column of $W^k$ with index $\ell_{k+1}$ such that, by swapping $\phi_{k+1}$ and $\phi_{\ell_{k+1}}$, the matrix $W^k_{[1:k+1,1:k+1]}$ is positive definite.
Thus, if $\textrm{tol}_1$ is sufficiently small, one can also prove convergence of OGRA analogously to GRA.

%%-------------------------------
\section{Numerical Results}\label{sec:numerics}
We test GRA and OGRA on the setting described in Sec.~\ref{sec:distribution}.
We choose a maximum control time of 160 ms, which corresponds to a normalized time $t_f=16$.
The shift of the parameter $\Delta$ is set to $\Delta_0=4\pi$ and the width of its interval to $4\pi\Delta_1$, with $\Delta_1=0.2$.
We consider two different probability distributions $P_\star$, a simple Gaussian one (see panel on the left in Fig.~\ref{fig:gauss_reconstructed}) and a step distribution with three peaks (see panel on the left in Fig.~\ref{fig:dirac_reconstructed}).
They are discretized by a uniform mesh of 100 points (10 points in each direction). Similarly, we discretize the set of linearly independent functions $\{\phi_j\}_{j=1}^K$ by setting $K=100$ and $\phi_j=\ve_j\in\mathbb{R}^{100}$ the $j$-th canonical vector in $\mathbb{R}^{100}$.
Finally, we fix the tolerances for OGRA to be $\textrm{tol}_1=10^{-14}$ and $\textrm{tol}_2=10^{-4}$.

Now, let us briefly discuss how we solve the sub-steps of the algorithms numerically.
The initialization and discriminatory step problems are solved by a second-order trust-region method.
For the fitting step, we use the equivalent compact form~\eqref{eq: fitting step W}.
The corresponding first-order optimality system is given by 
\begin{equation}\label{eq:fitting step optimality}
    [W^k]_{[1:k,1:k]}\beta=[W^k]_{[1:k,k+1]}.
\end{equation}
Since the matrix $[W^k]_{[1:k,1:k]}$ is symmetric and positive definite, any solution to Eq.~\eqref{eq:fitting step optimality} is a global solution to Eq.~\eqref{eq: fitting step W}.
Hence, we solve the fitting step problem by solving the linear system \eqref{eq:fitting step optimality} using a solver based on the SVD.
This solver first computes the SVD of $[W^k]_{[1:k,1:k]}$, i.e. two orthogonal matrices $U,V\in\mathbb{R}^{k\times k}$ and a diagonal matrix $\Sigma\in\mathbb{R}^{k\times k}$ such that $U\Sigma V^\top=[W^k]_{[1:k,1:k]}$.
To make the method more robust against numerical instabilities, it then removes all singular values that are smaller than a given tolerance, and the corresponding columns of $U$ and $V$.
Finally, it computes $\beta$ by setting $\widetilde{\beta}=V^\top\beta$ and solving $\Sigma\widetilde{\beta}=U^\top[W^k]_{[1:k,k+1]}$.

After running the algorithms, we reconstruct $P_\star$ by solving problem \eqref{eqbeta}.
Notice that, using the notation \eqref{eq: W}-\eqref{eq: gamma}, the gradient of the cost function in \eqref{eqbeta} is given by $W\beta-\sum_k\Gamma(\textbf{u}_k,t_k)^\top\textbf{Y}^{\textrm{exp}}_{\textbf{u}_k}(t_k)$,
where the columns of $\Gamma$ are given by the $\gamma_j(\textbf{u}_k,t_k)$ defined in Eq.~\eqref{eq: gamma}.
We can also immediately see that the Hessian of the cost function in Eq.~\eqref{eqbeta} is exactly $W$, which is guaranteed to be positive definite by our analysis in Sec.~\ref{sec:analysis}.
Hence, the global solution to Eq.~\eqref{eqbeta} is given by the (unique) solution to
\begin{equation}\label{eq:Reconstruction}
    W\beta=\sum_k\Gamma(\textbf{u}_k,t_k)^\top\textbf{Y}^{\textrm{exp}}_{\textbf{u}_k}(t_k).
\end{equation}
However, in order to ensure that the  coefficients of the computed solution  correspond to a probability distribution (i.e. belong to $\widehat{\mathbb{R}}^K$), we add the necessary constraints and solve Eq.~\eqref{eqbeta} with the second-order interior point algorithm of MATLAB's \textit{fmincon}-solver.
Nonetheless, the code includes an option to solve directly Eq.~\eqref{eq:Reconstruction} using a SVD solver (see Section \ref{sec:code}).

Now, we run both GRA and OGRA on the canonical set $\{\phi_j\}_{j=1}^{100}$ of hat functions.
In contrast to \cite{spinpaper}, we do not include any additional random vectors in the canonical set for OGRA and also do not remove any elements from the set during OGRA (but still reorder them).
The reason for this is that we experienced for the problem of this paper that additional random elements do not improve the results and removing elements from the canonical set does not reduce the number of controls, but is more likely to make the final identification problem numerically unstable.
While GRA computes $100$ controls, OGRA only designs $51$ by skipping 48 discriminatory steps.
We then choose $P_\star$ as the Gaussian distribution in Fig.~\ref{fig:gauss_reconstructed} (left) and compute the corresponding experimental realizations $\{\textbf{Y}^{\textrm{exp}}_{\textbf{u}_k}(t_k)\}_{k=1}^{\widetilde{K}}$ for the two resulting sets of control fields, with $\widetilde{K}=100$ for GRA and $\widetilde{K}=51$ for OGRA.
Reconstructing $P_\star$ as described above, we obtain the coefficient vectors $\beta_{rec}$ and thereby the distributions $P_{rec}=\sum_{j=1}^{100}\beta_{rec,j}\phi_j$ corresponding to GRA and OGRA, shown in Fig.~\ref{fig:gauss_reconstructed}.
Looking at the errors with respect to the true distribution $P_\star$ shown in Fig.~\ref{fig:Test1_err}, we observe that GRA  outperforms OGRA by one order of magnitude.
However, the difference is so small that it is not visible in the reconstructed distributions.
Similar results are obtained for a step distribution with three peaks in Fig.~\ref{fig:dirac_reconstructed}.
\begin{figure}[t]
	\centering
    \includegraphics[width=1\linewidth]{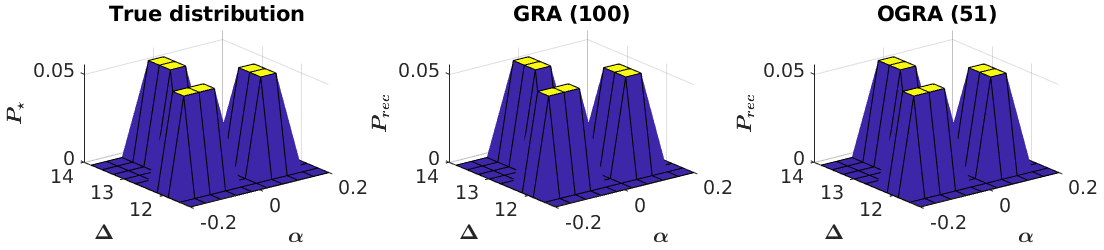}
    \caption{Same as Fig.~\ref{fig:gauss_reconstructed} but for a step distribution with three peaks. In brackets are the number of control fields for each set.}
    \label{fig:dirac_reconstructed}
\end{figure}

To investigate the dependence of the results on the choice of parameters, we repeat the experiment for different maximum control times, widths of the $\Delta$-interval and $K=400$ mesh points.
First, we take a look at the number of control fields generated by OGRA in Tab.~\ref{tab:OGRANumbers}.
\begin{table}[t]
\centering
    \begin{small}
	\begin{tabular}{|c||c|c|c|c||c|c|c|c|}
        \hline
        &\multicolumn{4}{c||}{\small{$K=100$}}&\multicolumn{4}{c|}{\small{$K=400$}}\\
		\hline \hline
		 \backslashbox[9mm]{$\Delta_1$}{$t_f$}& $8$  & $16$  & $24$  & $32$ & $8$  & $16$  & $24$  & $32$\\ 
		\hline \hline
        $0.1$& $\;$70$\;$ & $\;$\pmb{52}$\;$ & $\;$\pmb{50}$\;$ & $\;$\pmb{50}$\;$    &  259 & 299 & 276 &  240\\ 
        \hline
		$0.2$& $\;$\pmb{58}$\;$ & $\;$\pmb{51}$\;$ & $\;$\pmb{50}$\;$ & $\;$\pmb{50}$\;$    &  305 & 288 & \pmb{223} & \pmb{220}\\ 
        \hline
        $0.4$& $\;$\pmb{56}$\;$ & $\;$\pmb{50}$\;$ & $\;$\pmb{50}$\;$ & $\;$\pmb{50}$\;$    &  326 & 256 & \pmb{211} & \pmb{200}\\
        \hline
        $0.8$& $\;$\pmb{51}$\;$ & $\;$\pmb{50}$\;$ & $\;$\pmb{50}$\;$ & $\;$\pmb{50}$\;$    &  292 & \pmb{210} & \pmb{200} & \pmb{200}\\
        \hline
        $1.6$& $\;$\pmb{50}$\;$ & $\;$\pmb{50}$\;$ & $\;$\pmb{51}$\;$ & $\;$\pmb{50}$\;$    &  275 & \pmb{205} & \pmb{200} & \pmb{200}\\
        \hline
	\end{tabular} 
    \end{small}
 \caption{Number of controls computed by OGRA for a control bound $u_m=10$, and different numbers of discretization points $K$, maximum control times $t_f$ and widths $4\pi\Delta_1$ of the $\Delta$-interval. Bold numbers indicate that the number of OGRA controls is less than $60\%$ of the number of GRA controls. Notice that GRA always generates $K$ controls.}
 \label{tab:OGRANumbers}
\end{table}
We observe that the number of generated control fields is increasing with decreasing maximum control time and decreasing width of the $\Delta$-interval.
We also observe that the ratio between the number of GRA controls (which is equal to the number of mesh points $K$) and the number of OGRA controls is decreasing with an increasing number of mesh points.
To validate this point, we plot the number of controls for both algorithms, for different total numbers of mesh points in Fig.~\ref{fig:nr controls}.
\begin{figure}[t]
	\centering
	\includegraphics[width=0.3\linewidth]{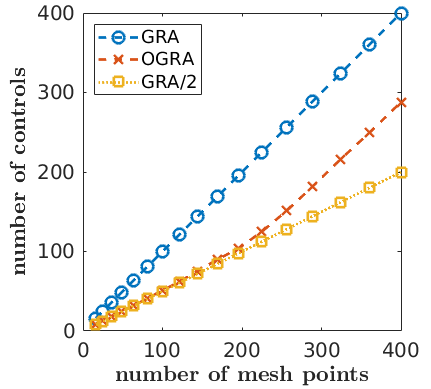}
	\caption{Number of controls for GRA (dashed circles) and OGRA (solid crosses) for different total numbers of mesh points. To highlight the ratio between the amount of controls, we also plot half the amount of GRA controls (dotted squares).}
	\label{fig:nr controls}
\end{figure}

An explanation of this behaviour is given by the condition number of the corresponding matrices $W$, defined in Eq.~\eqref{eq: W}, representing the compact form \eqref{eq: pbminW} of the final identification problem.
The condition numbers corresponding to GRA and OGRA for the settings in Tab.~\ref{tab:OGRANumbers} are shown in Tabs.~\ref{tab:u_m10_100} and~\ref{tab:u_m10}.
Based on our theoretical results for GRA and OGRA proving that $\widetilde{K}=K$ controls are sufficient, we also add a set of fully random controls (randomized within the given bounds $u_m$ and $t_f$) that has the same number of controls as GRA (i.e. $\widetilde{K}=100$ and $\widetilde{K}=400$, respectively).
\begin{table}[t]
\centering
    \begin{small}
	\begin{tabular}{|c||c|c|c|c||c|c|c|c||c|c|c|c|}
        \hline
        &\multicolumn{4}{c||}{GRA}&\multicolumn{4}{c||}{OGRA}&\multicolumn{4}{c|}{random control set}\\
		\hline \hline
		 \backslashbox[9mm]{$\Delta_1$}{$t_f$}& $8$  & $16$  & $24$  & $32$ & $8$  & $16$  & $24$  & $32$ & $8$  & $16$  & $24$  & $32$   \\ 
		\hline \hline
        $0.1$& $3e16$ & $\pmb{5e07}$ & $\pmb{7e03}$ & $\pmb{6e03}$     & $3e15$ & $\pmb{3e08}$ & $\pmb{7e05}$ & $\pmb{7e06}$    & $1e18$ & $\pmb{1e09}$ & $\pmb{1e06}$ & $\pmb{3e06}$ \\ 
        \hline
		$0.2$& $\pmb{7e09}$ & $\pmb{4e06}$ & $\pmb{3e03}$ & $\pmb{1e03}$    & $\pmb{5e09}$ & $\pmb{1e08}$ & $\pmb{1e08}$ & $\pmb{1e06}$     & $\pmb{7e12}$ & $\pmb{1e09}$ & $\pmb{4e04}$ & $\pmb{9e03}$\\
        \hline
        $0.4$& $\pmb{1e11}$ & $\pmb{5e03}$ & $\pmb{1e03}$ & $\pmb{1e03}$    & $\pmb{2e11}$ & $\pmb{3e06}$ & $\pmb{1e07}$ & $\pmb{1e07}$     & $1e16$ & $\pmb{7e04}$ & $\pmb{9e03}$ & $\pmb{1e03}$\\
        \hline
        $0.8$& $\pmb{1e06}$ & $\pmb{2e03}$ & $\pmb{1e03}$ & $\pmb{1e03}$     & $\pmb{3e07}$ & $\pmb{8e05}$ & $\pmb{2e05}$ & $\pmb{2e05}$    & $\pmb{7e10}$ & $\pmb{3e03}$ & $\pmb{2e03}$ & $\pmb{1e03}$\\
        \hline
        $1.6$ & $\pmb{2e04}$ & $\pmb{9e02}$ & $\pmb{1e03}$ & $\pmb{8e02}$     & $\pmb{2e06}$ & $\pmb{4e07}$ & $\pmb{1e05}$ & $\pmb{7e07}$     & $\pmb{5e09}$ & $\pmb{1e04}$ & $\pmb{2e03}$ & $\pmb{1e03}$ \\
        \hline
	\end{tabular} 
    \end{small}
 \caption{Condition number of $W$ for different control sets, maximum control times $t_f$ and widths $4\pi\Delta_1$ of the $\Delta$-interval. The total number of mesh points is $K=100$ and the bound on the control is $u_m=10$. Bold numbers indicate that the condition number is smaller than $1e15$.}
 \label{tab:u_m10_100}
\end{table}
\begin{table}[t]
\centering
    \begin{small}
	\begin{tabular}{|c||c|c|c|c||c|c|c|c||c|c|c|c|}
        \hline
        &\multicolumn{4}{c||}{GRA}&\multicolumn{4}{c||}{OGRA}&\multicolumn{4}{c|}{random control set}\\
		\hline \hline
		 \backslashbox[9mm]{$\Delta_1$}{$t_f$}& $8$  & $16$  & $24$  & $32$ & $8$  & $16$  & $24$  & $32$ & $8$  & $16$  & $24$  & $32$   \\ 
		\hline \hline
        $0.1$& $2e20$ & $1e19$ & $1e20$ & $2e15$ &     $3e19$ & $2e19$ & $9e19$ & $3e15$     & $4e19$ & $2e19$ & $2e19$ & $3e18$\\ 
        \hline
		$0.2$& $1e19$ & $1e19$ & $\pmb{1e14}$ & $\pmb{4e13}$ &     $6e19$ & $5e19$ & $\pmb{2e14}$ & $\pmb{1e14}$&     $2e19$ & $2e19$ & $6e19$ & $3e18$\\ 
        \hline
        $0.4$& $8e19$ & $2e19$ & $\pmb{1e14}$ & $\pmb{1e04}$&     $5e19$ & $8e18$ & $\pmb{1e14}$ & $\pmb{4e07}$ &     $6e19$ & $4e19$ & $8e15$ & $\pmb{2e06}$\\
        \hline
        $0.8$& $3e19$ & $\pmb{2e13}$ & $\pmb{1e04}$ & $\pmb{9e03}$ &     $3e19$ & $\pmb{2e14}$ & $\pmb{1e07}$ & $\pmb{5e07}$&     $5e19$ & $4e18$ & $\pmb{1e08}$ & $\pmb{1e05}$\\
        \hline
        $1.6$ & $6e20$ & $\pmb{1e10}$ & $\pmb{2e04}$ & $\pmb{6e03}$ &     $1e20$ & $\pmb{4e11}$ & $\pmb{8e07}$ & $\pmb{2e08}$ &     $3e19$ & $6e18$ & $\pmb{1e08}$ & $\pmb{2e04}$ \\
        \hline
	\end{tabular} 
    \end{small}
 \caption{Same as Tab.~\ref{tab:u_m10_100} but for a total number of mesh points $K=400$.}
 \label{tab:u_m10}
\end{table}
We observe that the condition number shows the same correlation with respect to the maximum control time, width of the $\Delta$-interval and number of mesh points, as the number of OGRA controls.
In particular, the condition number of OGRA is below $1e15$ for all settings where OGRA computed less than $60\%$ of the number of GRA controls.

Regarding the condition numbers, GRA and random controls show the same behaviour as OGRA.
The reason can be found by taking a closer look at the entries of the matrix $W$.
It can be shown that the difference between two adjacent rows or columns of $W$ is bounded in norm by $u_m$, $t_f$ and the mesh size for the probability distribution, i.e.,  $\alpha_{\ell+1}-\alpha_\ell$ and $\Delta_{\ell+1}-\Delta_\ell$.
The interested reader can find more details about this result in \ref{sec:appendixA}.
We conclude that, if the control bound, the maximum control time, or  the mesh size (or equivalently the width of the $\Delta$-interval) is too small, the difference between two adjacent rows/columns of $W$ can become numerically equal to zero, implying that $W$ has a nontrivial kernel.

In order to investigate the impact of this numerical instability on the reconstructed results, we consider again the setting of the beginning of this section (i.e. $\Delta_1=0.2$ and $t_f=16$), but for $K=400$ mesh points.
The results for a Gaussian and a step distribution with three peaks are plotted in Figs.~\ref{fig:gauss_reconstructed_400} and \ref{fig:3peaks_reconstructed_400}, respectively.
\begin{figure}[t]
	\centering
    \includegraphics[width=1\linewidth]{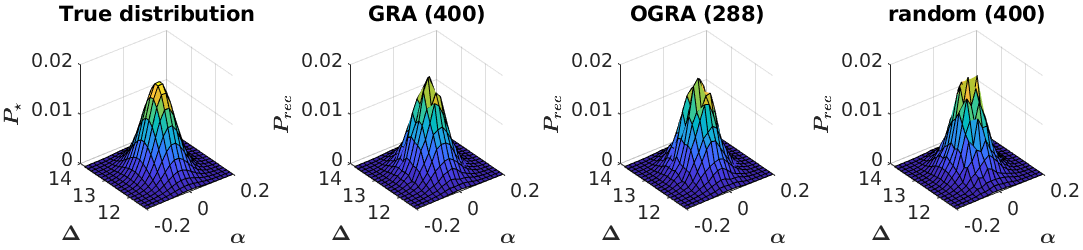}
    \caption{Same as Fig.~\ref{fig:gauss_reconstructed} but for $K=400$ and including the reconstructed distribution for $400$ random control fields. In brackets are the number of control fields for each set.}
    \label{fig:gauss_reconstructed_400}
\end{figure}
\begin{figure}[t]
	\centering
    \includegraphics[width=1\linewidth]{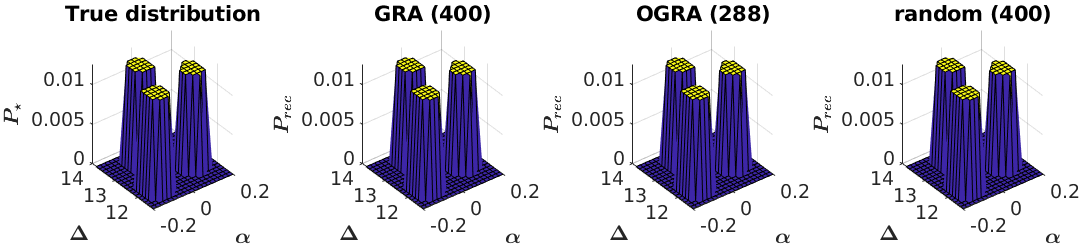}
    \caption{Same as Fig.~\ref{fig:gauss_reconstructed_400} but for a step distribution with three peaks. In brackets are the number of control fields for each set.}
    \label{fig:3peaks_reconstructed_400}
\end{figure}
We observe that all three control field sets are able to fully reconstruct the step distribution and, at least partially, the Gaussian distribution.
This is because the admissible set of solutions for the final identification problem \eqref{eqbeta} is restricted to $\widehat{\mathbb{R}}^K$.
Thus, a bad condition number does not necessarily imply that it is impossible to (at least partially) reconstruct the true probability distribution.
However, a good condition number guarantees stability of the numerical solver and improves the accuracy of the results.
In this context, notice that if we would sufficiently increase either the control bound $u_m$, or the maximum control time $t_f$, both GRA and OGRA would show better condition numbers and be able to perfectly reconstruct also the Gaussian distribution in Fig.~\ref{fig:gauss_reconstructed_400}.

We observe also that, if one knows the number of sufficient control functions $\widetilde{K}=K$, then even completely random control fields can be able to perform similarly to GRA and OGRA controls.
However, while OGRA finds automatically $\widetilde{K}$ (reduces the number of control fields to a sufficient amount), there is no indicator for a sufficient amount of random controls in general.
Additionally, the corresponding condition numbers are in many cases worse than for GRA and OGRA, as seen in Tabs.~\ref{tab:u_m10_100} and \ref{tab:u_m10}, meaning that they are more likely to show numerical instabilities.
Thus, the recommended strategy is clearly OGRA, since it is able to reduce the number of control fields by up to $50\%$, while accurately reconstructing the probability distributions.

Lastly, we remark that making the tolerance $\textrm{tol}_2$ smaller can generally lead to even fewer controls being computed by OGRA. However, this in turn can lead to less accurate results in the reconstructed solution, meaning the user has to decide for themselves if such a trade-off is desirable.

\section{Conclusion}\label{sec:conclusion}
In conclusion, we introduce \texttt{SPIRED}, a Greedy reconstruction algorithm to estimate spin distribution in NMR. We show that this approach can be used to jointly find the distribution of two Hamiltonian parameters, namely the offset term and the magnetic field inhomogeneity. We discuss the accuracy and limitations of this method through experimentally relevant numerical simulations. We provide and describe the codes allowing to reproduce the results of this paper. A proof of the algorithm convergence is also given.

This paper opens the way to a series of interesting questions in quantum control. A first step is to apply this algorithm to other areas in which an ensemble of quantum systems is used. Among others, we mention Bose Einstein Condensates in an optical lattice~\cite{anderson,BEC2021} or molecular rotational dynamics in gas phase~\cite{RMP:rotation,sugny2003}. The greedy reconstruction algorithm can in principle be applied to these examples, but specific constraints related to the corresponding experimental setups would be to take into account and would require adaptations of the \texttt{SPIRED} code. A final stage concerns the experimental implementation of this approach which seems realistic in the near future in view of the current state of the art.

\section*{Acknowledgements}
Simon Buchwald is funded by the DFG via the collaborative research center SFB1432,
Project-ID 425217212. 
Gabriele Ciaramella is member of the INDAM GNCS. The research of D. Sugny has been supported by the ANR project ``QuCoBEC'' ANR-22-CE47-0008-02.

\appendix
\section{Numerical stability of the matrix $W$}
\label{sec:appendixA}
We study in this section the numerical stability of $W$. Notice that for constant controls the solution to the dynamical equation~\eqref{eq1} can be written as 
\begin{equation*}
    \textbf{X}_{\textbf{u},(\Delta,\alpha)}(t)=e^{t(\Delta A+\alpha(\textbf{u}_xB_x+\textbf{u}_yB_y))}\textbf{X}_0.
\end{equation*}
Recall that for two matrices $X$ and $Y$, we have 
\begin{equation*}
    \|e^{Y}-e^X\|\leq\|Y-X\|e^{\|Y\|}e^{\|X\|}.
\end{equation*}
Now, consider two parameter pairs $(\alpha_\ell,\Delta_\ell)$ and $(\alpha_{\ell+1},\Delta_{\ell+1})$, and define $D_{\ell}:=t(\Delta_{\ell} A+\alpha_{\ell}(\textbf{u}_xB_x+\textbf{u}_yB_y))$ and $D_{\ell+1}:=t(\Delta_{\ell+1} A+\alpha_{\ell+1}(\textbf{u}_xB_x+\textbf{u}_yB_y))$.
Since $\|\textbf{X}_0\|=1$, $|\textbf{u}_x|\leq u_m$ and $|\textbf{u}_y|\leq u_m$, we obtain
\begin{align*}
    \|\textbf{X}_{\textbf{u},(\alpha_{\ell},\Delta_{\ell})}(t)-\textbf{X}_{\textbf{u},(\alpha_{\ell+1},\Delta_{\ell+1})}(t)\|&\leq e^{\|D_{\ell}\|}e^{\|D_{\ell+1}\|}\|t(\Delta_{\ell} A+\alpha_{\ell}(\textbf{u}_xB_x+\textbf{u}_yB_y))\\
    &\quad-t(\Delta_{\ell+1} A+\alpha_{\ell+1}(\textbf{u}_xB_x+\textbf{u}_yB_y))\|\\
    &\leq e^{\|D_{\ell}\|}e^{\|D_{\ell+1}\|}t_f\Big(|(\Delta_{\ell}-\Delta_{\ell+1})|\|A\|\\
    &\quad + |\alpha_{\ell}-\alpha_{\ell+1}|u_m(\|B_x\|+\|B_y\|)\Big).
\end{align*}
Since the exponential mapping is continuous, we have $e^{\|D_{\ell+1}\|}\rightarrow e^{\|D_{\ell}\|}$ for $\Delta_{\ell+1}\rightarrow\Delta_{\ell}$ and $\alpha_{\ell+1}\rightarrow\alpha_{\ell}$.
Thus, the norm of the difference between the two solutions $\textbf{X}_{\textbf{u},(\alpha_{\ell},\Delta_{\ell})}(t_f)$ and $\textbf{X}_{\textbf{u},(\alpha_{\ell+1},\Delta_{\ell+1})}(t)$ is bounded by the differences $|\Delta_{\ell}-\Delta_{\ell+1}|$, $|\alpha_{\ell}-\alpha_{\ell+1}|$, the bound to the control $u_m$ and the maximum control time $t_f$.
Recalling \eqref{eq: gamma} and that $\phi_j=\ve_j$ in our example, the matrix entries of $W$ are given by
$$W_{\ell,j}=\sum_k\langle \textbf{Y}_{\textbf{u}_k,(\alpha_\ell,\Delta_\ell)}(t_k)|\textbf{Y}_{\textbf{u}_k,(\alpha_j,\Delta_j)}(t_k)\rangle.$$

 \bibliographystyle{elsarticle-num} 
 \bibliography{SPIRED}





% \end{thebibliography}
\end{document}